\documentclass[oneside]{amsart}
\usepackage {CJKutf8}
\usepackage{amssymb}
\usepackage{amsthm}
\usepackage[abbrev]{amsrefs}
\usepackage{hyperref}

\usepackage{enumitem}

\usepackage{todonotes}

\usepackage{float} % to fix tables

\usepackage{geometry}
\theoremstyle{plain}%
\newtheorem{theorem}{Theorem}%  meant for continuous numbers
\newtheorem{proposition}[theorem]{Proposition}% 
\newtheorem{lemma}{Lemma}%  meant for continuous numbers
\newtheorem{corollary}{Corollary}%  meant for continuous numbers

\theoremstyle{remark}%
\newtheorem{remark}{Remark}%

\theoremstyle{definition}%
\numberwithin{equation}{section}
\numberwithin{theorem}{section}
\numberwithin{lemma}{section}
\numberwithin{corollary}{section}

\newcommand{\inn}[2]{\left\langle#1,\,#2\right\rangle}

\DeclareMathOperator{\supp}{supp}

\DeclareMathOperator{\dist}{dist}

\newcommand{\Lap}{\Delta}
\newcommand{\di}{\partial}

\newcommand{\br}[1]{\left\langle#1\right\rangle}
\newcommand{\si}{\sigma}
\newcommand{\eps}{\epsilon}

\newcommand{\g}{\gamma}
\newcommand{\al}{\alpha}

\newcommand{\Cb}{\mathbb{C}}

\newcommand{\Zb}{\mathbb{Z}}
\newcommand{\Rb}{\mathbb{R}}

\newcommand{\one}{\ensuremath{\mathbf{1}}}

\newcommand{\hf}{\ensuremath{\mathfrak{h}}}

\newcommand{\Ga}{\Gamma}

\renewcommand{\l}{\lambda} % renewed from slash l
\newcommand{\abs}[1]{\ensuremath{\left\lvert#1\right\rvert}}
\newcommand{\norm}[1]{\ensuremath{\left\lVert#1\right\rVert}}
\newcommand{\sbr}[1]{\left[#1\right]}
\newcommand{\Set}[1]{\left\{#1\right\}}

\newcommand{\md}[6]{\ensuremath{
		\ifinner
		\tfrac{\partial{^{#2}}#1}{\partial{#3^{#4}}\partial{#5^{#6}}}
		\else
		\tfrac{\partial{^{#2}}#1}{\partial{#3^{#4}}\partial{#5^{#6}}}
		\fi
}}
\newcommand{\del}[1]{\Bigl(#1\Bigr)}
\newcommand{\thmref}[1]{Theorem~\ref{#1}}

\newcommand{\secref}[1]{Section~\ref{#1}}
\newcommand{\lemref}[1]{Lemma~\ref{#1}}
\newcommand{\propref}[1]{Proposition~\ref{#1}}
\newcommand{\remref}[1]{Remark~\ref{#1}}
\newcommand{\figref}[1]{Figure~\ref{#1}}
\newcommand{\corref}[1]{Corollary~\ref{#1}}

\usepackage[normalem]{ulem}
\usepackage{soul}
\setstcolor{blue}

\definecolor{green}{rgb}{0.0, 0.5, 0.5}

\definecolor{lgray}{gray}{0.9}
\definecolor{llgray}{gray}{0.95}
\definecolor{lllgray}{gray}{0.975}

\newcommand{\cD}{\mathcal{D}}

\newcommand{\cS}{\mathcal{S}}

\newcommand{\cX}{\mathcal{X}}
\newcommand{\cY}{\mathcal{Y}}

\newcommand{\nc}{\newcommand}

\nc{\h}{\delta}
\nc{\G}{\Gamma}
\nc{\et}{\eta} 
\nc{\gam}{\gamma}
\nc{\ka}{\kappa}
\nc{\lam}{\lambda}
\nc{\Lam}{\Lambda}
\nc{\ta}{\tau}
\nc{\w}{\omega}
\nc{\io}{\iota}
\nc{\s}{\sigma}
\nc{\vphi}{\varphi}

\nc{\e}{\epsilon}

\nc{\ran}{\rangle}
\nc{\lan}{\langle}

\renewcommand{\Re}{\mathrm{Re}} % Real part
\renewcommand{\Im}{\mathrm{Im}} % Imaginary part

\nc{\bfone}{{\bf 1}}
\nc{\1}{{\bf 1}}

\nc{\dd}{\mathrm{d}}
\newcommand{\DETAILS}[1]{}

\DeclareMathOperator{\Ad}{ad}
\newcommand{\ad}[3]{\Ad^{#1}_{#2}(#3)}

\newcommand{\Rem}{\mathrm{Rem}}

\newcommand{\Norm}[1]{{\left\vert\kern-0.25ex\left\vert\kern-0.25ex\left\vert #1 
		\right\vert\kern-0.25ex\right\vert\kern-0.25ex\right\vert}}

\newcommand{\A}{\mathcal{A}}

\begin{document}
 
\title[Dynamical localization]{Upper bounds in non-autonomous quantum dynamics 
}

\author{Jingxuan Zhang}
\address{Yau Mathematical Sciences Center\\
	Tsinghua University\\
	Haidian District\\
	Beijing 100084, China }
\email{jingxuan@tsinghua.edu.cn}

\date{\today}

\pagestyle{plain}

\begin{abstract}
	We prove upper bounds on outside probabilities for generic non-autonomous Schr\"odinger operators on lattices of arbitrary dimension. Our approach is based on a combination of commutator method originated in scattering theory and novel monotonicity estimate for certain mollified asymptotic observables that track the spacetime localization of evolving states. Sub-ballistic upper bounds are obtained, assuming that momentum vanishes sufficiently fast in the front of the wavepackets. A special case  gives a refinement of the general ballistic upper bound of Radin-Simon's, showing that the evolution of wavepackets are effectively confined to a strictly linear light cone with explicitly bounded slope. All results apply to long-range Hamiltonian with polynomial decaying off-diagonal terms and can be extended, via a frozen-coefficient argument, to generic nonlinear  Schr\"odinger equations on lattices. 
%		In particular, our results apply to long-range Hamiltonian with polynomial decaying off-diagonal terms can.
%Similar to the classical result of Radin-Simon's, the specific form of potential plays no role in our analysis, as it does not contribute to the commutators involved. 

	%	Thus our results apply to long-range Hamiltonian with polynomial decay, and are valid with no special structural assumptions on the time-dependent potential. 

%	In particular, our results imply for  nonlinear  Schr\"odginer equations that typically arise as the mean-field limits of bosonic quantum many-body systems, dynamical spreading of the wavepacket is effectively confined in a strictly linear effective light cone, in contrast to the situation in the full many-body dynamics for which the effectively light cone slope typically grows in time. 

%	$p$-th order diffusive norm $\sum \abs{j}^p\abs{q_j(t)}$  of any solution $q(t)$ with well-localized initial data grows at most as $O(t^{dp})$. In particular, for $d=1$, there is no super-ballistic transport in general nonlinear, non-autonomous quantum dynamics. The proof is based on a direct upper bound on the outsider probabilities, which shows that the wavepacket decays rapidly outside a linearly spreading region. 
\end{abstract}

\maketitle

	\section{Introduction}

We consider the following  non-autonomous Sch\"odinger equation on $\hf:=\ell^2(\Zb^d),\,d\ge1$:
\begin{align}
	\label{SE}
	i\di_t u  = (H_0 + V(t))u.
\end{align}
Here $H_0$ is a bounded self-adjoint operator on $\hf$ and $V(t)$ is a time-dependent potential.
Assuming only that $V(t)$ is uniformly bounded in $\ell^\infty$ so that \eqref{SE} is globally well-posed on $\hf$, we seek upper bound on the outside probabilities of wavepackets evolving according to \eqref{SE} without additional structural on the potential. To avoid distraction, in what follows we always normalized the $\ell^2$-norm of solutions to \eqref{SE} to one.

The lack of control on $V(t)$ suggests the use of commutators in which the potential does not contribute. A natural candidate is the momentum  operator 
$$
	A=i[H_0,\abs{x}],
$$
first defined as a quadratic form on the natural domain $\cD:=\Set{u\in\hf:\norm{\abs{x}u}<\infty}$.
Clearly, the choice of $V(t)$ has no bearing on $A$. 
Suppose the kernel of $H_0$ satisfies suitable decay condition so that $A$ extends to a bounded operator on all of $\hf$. Then, at the level of expectation, the classical argument of Radin-Simon \cite{RS} yields a general ballistic upper bound:
\begin{align}
	\label{RSest}
	\norm{\abs{x}u_t}\le \norm{\abs{x}u_0}+C t  ,
	\quad t>0.
\end{align}
See also \cite{BdMS}*{Appd.~A.2} for extensions to higher moments.
With no additional structural assumption on $V(t)$, estimate \eqref{RSest} is the only constraint that we are aware of on the dynamical spreading of wavepackets evolving according to such a general equation as \eqref{SE}.

The Radin-Simon upper bound implies, via Markov's inequality, that at the level of probability distribution, solutions to \eqref{SE} with well-localized initial state must decay in the region $\Set{\abs{x}\ge C  t^\al}$ for any $\al>1$. We seek to further quantify the outside probabilities for solutions $u_t,\,t\ge0$ to \eqref{SE}, defined as
\begin{align}
	\label{Pdef}
	P(N,t):=\sum_{\abs{x}> N}\abs{u_t(x)}^2. 
\end{align}
For autonomous Schr\"odinger operator $H=H_0+V$, a  natural first step is to invoke resolvent estimate together with explicit integral representation of the propagator $e^{-itH}$. Indeed,  for the solution to \eqref{SE} generated by an initial state $u_0\in\hf$, the outside probabilities $P(N,t)$  can be conveniently written in terms of the resolvent $R(z):=(z-H)^{-1}$ via Dunford functional calculus as
\begin{align}
	\label{Pdecomp}
	P(N,t)=&\sum_{x\in\supp u_0}\sum_{\abs{y}>N}\abs{u_0(x)}^2\abs{\inn{e^{-itH}\delta_x}{\delta_y}}^2,\\
	\inn{e^{-itH}\delta_x}{\delta_y}=&\frac{1}{2\pi i} \oint_\Ga e^{-itz}{\br{R(z)\delta_x,\delta_y}}\,dz.\label{Dunford}
\end{align}
Here $\Ga$ is any positively oriented contour around the spectrum $\si(H)$.  
Through the integral representation \eqref{Dunford}, upper bounds on $P(N,t)$ are naturally tied to resolvent estimates for $R(z)$, whence spectral properties of the Hamiltonian enter. This is the starting point for many fruitful research on dynamical upper bounds for Schr\"odinger operator, see for instance the general approach of  Damanik--Tcheremchantsev \cites{DT,DTa,GKT} via transfer matrix on $\Zb$ and, more recently, Shamis-Sodin \cite{SSe} for upper bounds in higher dimensions; see also \cites{ZZ,GK}. 
An excellent review on the state of the art can be found in Damanik-Malinovitch-Young \cite{DMY}.

Along this route, a key ingredient is the Combes–Thomas estimate for the resolvent. 
We recall the following formulation  due to Aizenmann \cite{Aiz}*{Lem.~II.1}: 
\begin{lemma}[Combes–Thomas estimate]
	 Take $H=H_0+V$ with $\abs{H_0(x,y)}\le B\exp(-m\delta(\abs{x-y}))$, where $m,\,B>0$ and $\delta(\cdot)$ satisfies, for some $\al,S>0$,
	 \begin{align}
	 	\label{deltaDef}
	 	\int_1^\infty r^{d-1} e^{-(m-\al)\delta(r)}\,dr\le S.
	 \end{align} 
	 Then for any $b>1$, there exists $v=v(\al,B,S,b)>0$ s.th.  
	 \begin{align}
	 	\label{CTest}
	 \abs{\br{R(z)\delta_x,\delta_y}}\le  2e^{- \delta(\abs{x-y})/v},\quad 1\le \dist (z,\si(H))<b.
	 \end{align}
\end{lemma}
To motivate our main results below, we recall how the Thomas-Combes is used to show that outside probabilities are exponentially suppressed away from a strictly linear effective light cone.  
Indeed, choosing the contour in \eqref{Dunford} to be the boundary of the box $\Set{\abs{\Im z}\le 1,\abs{\Re z} \le \norm{H}+1}$, we conclude from \eqref{CTest} that 
for some $C=C(\norm{H})>0$ and $v=v(\al,B,S,\norm{H})>0$,
\begin{align}
	\label{17}
		\abs{\inn{e^{-itH}\delta_x}{\delta_y}}\le Ce^{t-\delta(\abs{x-y})/v}.
\end{align}
Plugging \eqref{17} back to \eqref{Pdecomp} and setting $R:=\sup\Set{\abs{x}:x\in\supp u_0}$, we find that
% there exists $t_0=t_0(R, \delta)>0$ s.th.
\begin{align}
	\label{bench}
\text{	$P(N,t)\le C  e^{- N}$ \quad for \quad 
%	all $t\ge t_0$ and 
%	$N=N(t)$ satisfying 
	$\delta(N-R)> 2vt$.} 
\end{align}
If $u_0$ has compact support and $\delta(r)=r$, then \eqref{bench} yields exponential decay away from the effective light cone $\Set{\abs{x}\le R+ 2vt}$. However,   limitations of this approach based on resolvent estimate arise when
\begin{itemize}
	\item  $u_0$ does not have compact support: In this case  $R=\infty$ in \eqref{bench}.
	\item $H_0$ has polynomially decaying off-diagonal entries, i.e.  $\delta(r)=\log(1+r)$ in \eqref{deltaDef}: In this case \eqref{17} decays only when $\abs{x-y}=O(e^t)$.
	\item The potential is time-dependent: In this case the propagator cannot be written as \eqref{Dunford}. 
\end{itemize}
It is not obvious how to extend the existent method to obtain upper bounds on the outside probabilities in these settings. The present paper aims to supply some new approaches to cover  parts of this gap.

\subsection*{Notation}\label{secNota}
We denote by $\1_X$ the characteristic functions associated to a set $X$, $\dist_X(x):=\inf_{y\in X}\abs{x-y}$ the distance function to $X$,  $\cD(A)$ the domain of an operator $A$, $B_a =\Set{x\in\Zb^d:\abs{x}\le a}$ the ball of radius $a$ around the origin,   and  
$\|\cdot\|$   the norm of  operators on $\hf$ and sometimes that of vectors in $\hf$.
%For a bounded operator $A$ on $\hf$, we denote by $A_{xy},\,x,y\in \Lam$, the operator kernel (matrix) of $A$. 
%To shorten notations, we will drop the dependence of various operators on the full domain $\Lam$ when no confusion arises (see e.g.~\eqref{1.1}).  
%$\cB(\hf)$ denotes the space of bounded operators on $\hf$.
We make no distinction in our notation between a function $f\in\hf$ and the associated multiplication operator $\psi(x)\mapsto f(x)\psi(x)$ on $\hf$. 
The commutator $[A,B]$ of two operators $A$ and $B$ is first defined as a quadratic form on $\cD(A)\cap\cD(B)$ (always assumed to be dense in $\hf$) and then extended to an operator. Similarly, the multiple commutators of $A$ and $B$ are defined recursively by $\ad{0}{B}{A}=A$ and $\ad{p}{B}{A}=[\ad{p-1}{B}{A},B]$ for $p=1,2,\ldots$.
%\end{remark}

\subsection{Results} 
In the remainder of this paper we consider generic non-autonomous Schr\"odinger operators $H(t)=H_0+V(t)$, where $H_0$ is  long-range self-adjoint operator with power-law decay, satisfying, for some integer $n\ge1$,  
\begin{align}
	\label{Lcond}
	M:=\sup_{x\in\Zb^d}\sum_{y\in\Zb^d}\abs{H_0(x,y)} \abs{x-y}^{n+1}<\infty.
\end{align} 
This condition shows that momentum operator $A=i[H_0,\abs{x}]$ is bounded. Indeed, set
\begin{align}
	\label{Kdef}
	\kappa:=\sup_{x\in\Zb^d}\sum_{y\in\Zb^d}\abs{H_0(x,y)} \abs{x-y}.
	%	\sup_{j\in\Zb^d}\sum_{j'\in\Zb^d} \abs{H_0(j,j')}\abs{j-j'}.
\end{align}
By the reveres triangle inequality and Schur's test, we have $\norm{Au}\le \kappa\norm{u}$ for $ u\in\cD(\abs{x})$. Since $\kappa\le M$, it follows that $A$ extends to a bounded operator on $\hf$ under condition \eqref{Lcond}. In particular, the Radin-Simon bound \eqref{RSest} holds. Consequently,  by Markov's inequality, the outside probability decays for  $t\to\infty$ and $N\ge Ct^\al,\,\al>1$. 

Our first result provides an upper bound on $P(N,t)$ with $N\ge Ct^\al$ and $\frac12<\al\le1$. The condition is, roughly speaking, that the momentum is vanishing sufficiently fast in the spacetime region corresponding to the `propagation front', $\Set{(x,t):\abs{x}\sim t^\al, t>0}$. 
More precisely, we prove the following:
\begin{theorem}
	\label{thm1}

	Let \eqref{Lcond} hold with $n\ge1$. Let $\frac12<\al\le1$, $\delta>0$, $v>\frac{\delta}{\al}$, and $\bar v =  v-\frac{\delta}{2\al} $.
	Let $W\in C_c^\infty(\Rb)$ be a smooth cutoff function with $\supp W\subset (0,1)$ and $W\equiv 1$ on $ (\frac14,\frac34)$. Set $W_t(\phi )=W(\frac {2\al }{\delta t^\al} (\phi -\bar vt^\al ))$. 
	
	Assume    $u_t,\,t\ge0 $ is a solution to \eqref{SE} satisfying, for some $R\ge0$,
	\begin{align}
		\label{uCond'}
		\norm{i[H,\dist _{B_R}]W_t(\dist _{B_R})u_t}\le   (v\al-\delta)t^{\al-1}\norm{W_t(\dist _{B_R})u_t},\quad t>0.
	\end{align}
	Then there exists $C=C(n,M,\delta,\al,W)>0$   s.th.~for all $t>0$,
	\begin{align}
		\label{propGen}
		P((vt)^\al+R,t)\le 	(1+Ct^{-\g}){P(R,0)} + C{t}^{-\beta},\quad \g:=2\al-1,\ \beta:=\min\del{(n+1)\g,(n+1)\al -1}.
	\end{align}

\end{theorem}
This theorem is proved in \secref{secPfthm2}. 

\begin{remark}\label{remN}
	Condition \eqref{Lcond} holds {uniformly for all $n\ge1$}	for operators $H_0$ with finite interaction range, e.g.~the discrete Laplacian $$
		[\Lap u](x)= \sum_{y\sim x} u(y),
$$
	as well as operators with exponentially decaying kernel, i.e. 
$$
		\sup_{x\in\Zb^d}\sum_{y\in\Zb^d} \abs{H_0(x,y)}e^{\abs{x-y}}<\infty .
$$
Alternatively, condition \eqref{Lcond} can be characterized by the norm-differentiability of the operator $e^{-i\xi\cdot x}H_0e^{i\xi\cdot x}$ for $\xi\in\Cb^d$, c.f.~\cites{CJWW,SWb}.
\end{remark}

\begin{remark}

For fixed $\al$, $n$, and cutoff function $W$, the prefactor $C\to\infty$ in \eqref{propGen} as either  $M\to\infty$  or $\delta\to0+$. If $\delta$ is not too small ($\delta\ge2\al$ will do), then the prefactor can be made independent of $\delta$. The precise dependence on the parameters can be tracked  down from \eqref{Cdef} below.
\end{remark}
The proof of \thmref{thm1} is based on commutator estimates, in which the potential $V(t)$ in \eqref{SE} does not contribute to begin with. Consequently, our main estimate \eqref{propGen} is independent of the potential. The full time-dependent Hamiltonian enters only through condition \eqref{uCond'} via a given solution. To better understand this, we note that a necessary condition for \eqref{uCond'} is 
$$\lim_{t\to\infty}\norm{i[H,\dist _{B_R}]W_t(\dist _{B_R})u_t}=0.$$ 
Simple geometric consideration shows that $W_t(\dist_R)$ amounts to a cutoff function associated to the spacetime region
$$
F_t:=	\Set{(x,t)\in\Zb^d \times \Rb_{>0}:( \tfrac14 v + \tfrac 34\bar v) t^\al+R\le \abs{x}\le (\tfrac34 v + \tfrac 14\bar v)t^\al+R}.
$$
Thus \eqref{uCond'} requires, in particular, that the region $F_t$ cannot contain any free wave. Put differently, we need the dynamics to push the solution to zero frequency as $t\to\infty$ in the region $F_t$.

It is in general a difficult to determine the origin of decay for \eqref{uCond'}. This is not unexpected if one views \eqref{SE} as a frozen-coefficient equation arising from nonlinear Schr\"odinger equations, as pointed out in recent works of Liu-Soffer and Soffer-Wu \cites{LS,SW,SWa}.
%,  typical examples obeying \eqref{uCond'} are self-similar solution of the form $u_t=e^{i\l t}\phi$  where $\phi$ is a ground state of \eqref{SE}. In this case \eqref{uCond'} becomes a condition on $\phi$ and can be verified for $\phi(x)$ satisfying suitable decay estimates (both in $\ell^2$-norm and in the norm induced by the momentum operator) as $\abs{x}\to\infty$. 
Determining further sufficient condition for \eqref{uCond'} in terms of the initial state and the potential $V(t)$ is important for applications of \thmref{thm1} to concrete Hamiltonians. This is of independent interest in itself and so we will pursue it elsewhere.

In the meantime, ballistic upper bounds (i.e., $\al=1$) on outside probabilities can be derived from \thmref{thm2} independent of the solutions. Indeed, we prove the following:
\begin{theorem}\label{thm2}
Let \eqref{Lcond} hold with $n\ge1$. Then for any $v>\kappa$ and $\delta:=v-\kappa$,
there exists $C=C(n,M,\delta)>0$ s.th.~for any solution to \eqref{SE},
\begin{equation}\label{propEst'}
		P(vt+R,t)
		\le  (1+Ct^{-1}){P(R,0)} + C{t}^{-n}   ,\quad R\ge0, t>0.
	\end{equation}  
\end{theorem}

\begin{proof}[Proof of \thmref{thm2} assuming \thmref{thm1}]
	
We apply \thmref{thm1} with $\al=1$. In this case, condition \eqref{uCond'} is time-independent and holds for any solution $u_t$ to \eqref{SE}, owning to the operator norm bound $v-\delta=\kappa \ge \norm{i[H,\dist_{B_R}]}$. (The second inequality is a consequence of Schur's test; see \lemref{lemA0}.)
	Since $\g=1$ and $\beta=n$ for $\al=1$, the desired estimate \eqref{propEst'} follows from  \eqref{propGen}. 
\end{proof}

%Extending \eqref{propGen}  with $\al\in(0,1)$ to all $t>0$ would require delicate conditions on the evolving states. In view of recent results on nonlinear scattering in the continuum, we expect that \eqref{propGen} holds with $\al\in(\frac12,1)$ if the solution is asymptotically given by a linear combination of weakly localized states, in the sense of \cites{SW,SWa,LS,SSf}.  As the present paper is concerned with generic dynamical behaviour of \eqref{SE}, 

We now   turn to illustrate some simple applications of the theorems above.

First, we note that if $H_0$ satisfies the decay condition \eqref{Lcond} \textit{uniformly for all $n\ge1$}, (c.f.~\remref{remN}), then applying \thmref{thm1} successively with $n=1,2,\ldots$ shows that  the decay of wavepackets away from the propagation region $\Set{\abs{x}\le R+vt^\al}$ is faster than polynomial of any degree. 
More precisely,  following \cites{GKT,DT,DTa} we define the decay parameter
\begin{align}
	\label{}
	S^+(\al):=-\limsup_{t\to\infty} \frac{\log P(t^\al-1,t)}{\log t },
\end{align}
and the upper transport exponent 
\begin{align}
	\label{}
	\al_u^+:=\sup\Set{\al:S^+(\al)<\infty}.
\end{align}
Then it follows directly from \thmref{thm1} that
\begin{corollary}
	If  \eqref{Lcond} holds {uniformly for all $n\ge1$}, then for any solution satisfying \eqref{uCond'} with $\frac12<\al\le1$, we have
	\begin{align}
		\label{}
		\al_u^+   \le \al.
	\end{align}
\end{corollary}

Next, we prove an upper bound on higher moments of the position operator away from the propagation region. For $r\ge 0$, let
\begin{align}
	\label{HpDef}
	P_r(N,t):=  \sum_{\abs{x}>N} \abs{x}^{r}\abs{u(x,t)}^2. 
\end{align}
Then $P(N,t)\equiv P_0(N,t)$. In general, $P_r,\,r>0$ is naturally related to the Sobolev/diffusive norms; see e.g. \cites{BW,ZZ}.
We prove the following:
%converse of Guarneri-Combes-Last theorem \cite{Las} for \eqref{NLS}
\begin{theorem}\label{thm3}
	Assume \eqref{propEst'} holds  with $n>dr_0/2$ for some $r_0> 0$, and assume $u_t,\,t\ge0$ solves \eqref{SE} with   initial state  satisfying 
	\begin{align}
		\label{q0Suff}
		P_{r_0}(0,0)\le B.
	\end{align}
	Then for any $0\le r<r_0$, there exists $C=C (n,M,B,d, v,r,r_0 )>0$ such that for all  $t\ge1$, 
	\begin{align}
		\label{MainEst}
		P_r(2vt,t)\le C. 
	\end{align}
\end{theorem}
This theorem is proved in \secref{secPfBall}. 

\begin{remark}
	The generalized Radin-Simon estimate (see e.g.~\cite{BdMS}*{Thm.~A.2}) asserts that  for initial state satisfying $P_r(0,0)\le C$, the full Sobolev norm $P_r(0,t)$ typically grows polynomially as $O(t^r)$. In view of this, estimate \eqref{MainEst} shows  that the growth contribution away from a linearly spreading region is at most $O(1)$. 
\end{remark}

\begin{remark}
	Conversely, if \eqref{MainEst} holds, then by Markov's inequality we have $P(2vt,t)\le Ct^{-r}$, c.f.~\eqref{propEst'}.
\end{remark}

\begin{remark}
	For compactly supported initial state,	condition \eqref{q0Suff} holds 	for all $r_0>0$. Thus estimate \eqref{MainEst} holds for all $r$ if $H_0$ is as in \remref{remN} and $u(0)$ has compact support. This is the setting for many results in the literature where $H_0$ is  the discrete Laplacian and $u_0$ is the point mass at the origin. 
	%	For the significance of this, see \cite[Thm.~4.1]{GKT}.
\end{remark}
%we see that for general nonlinear non-autonomous Schr\"odinger equation on $\Zb^d$, the $p$-th moment bound remains finite for all time \cjz{Check if this is non-trivial}

% in terms of the `diffusive norms' 
%(c.f.~\cite{BW}) 

%	To avoid distraction, in this paper we will not further discuss the well-posedness issue for \eqref{NLS}. 

%\cjz{Reformulate this only for the interior region and assert constant!}

%\end{remark}

%The information of the full flow enters \thmref{thm1} only through the validity threshold \eqref{Tdef}. Thus if $N(t)$ is at least of linear growth  with rate strictly greater than $\kappa$, \thmref{thm1} holds with  $T=\infty$ \textit{for any solution to \eqref{SE}}.  

%\begin{remark}
% As far as we are aware, \eqref{RSest} is the only ballistic upper bounds for \textit{generic non-autonomous Schr\"odinger equations}. 

%\end{remark}

Lastly, we note that similar to the Radin-Simon ballistic upper bounds, the potential $V(t)$ does not enter our estimates. 
%	Thus we understand Thms.~\ref{thm2}--\ref{thm1} as asserting that generic  non-autonomous Schr\"odinger equations on lattice obey the same  upper bound on outsider probabilities \eqref{propEst'}, \eqref{propGen}  as does the free evolution $ i \di_t u = H_0u$.  
Owning to this generality, by a standard frozen-coefficient argument, our results above extend to generic non-linear, non-autonomous Schr\"odinger equations of the form
\begin{align}
	\label{NLS}
	i\di_t q = \del{H_0 +\bar V(t) +\abs{N(t,q)} }q.
\end{align}
To avoid distraction, for our purpose we \textit{assume} global well-posedness for \eqref{NLS} on $\hf$. Then we have the following:
\begin{theorem}\label{thmNLS}
	Let $H_0$ satisfy \eqref{Lcond}, potential $\bar V(t)$  be uniformly bounded, and  nonlinearity obey  $\abs{N(t,q)}\le C_1$ for $t\ge0,\,\abs{q}\le C_2$.  
	% In fact, as far as Thms.~\ref{thm2}--\ref{thm3} are concerned,  detailed structure of  $\bar V(t)$ and $N(t,q)$ is totally not used. 
	Then the statements of Thms.~\ref{thm1}--\ref{thm3} are valid for solutions to \eqref{NLS}.
\end{theorem}
This theorem is proved in \secref{secPfNLS}.
The conditions on $V(t)$ and $N(t,q)$ can be further relaxed, so long as \eqref{NLS} and  the frozen-coefficient equation (see \eqref{NLSf} below) are globally well-posed on $\hf$.

Finally, we note that all results above can be extended to general geometry, with the outside probabiliries defined as $P_X(N,t)=\sum_{\dist_X(x)> N}\abs{u_t(x)}^2$ for any $X\subset\Zb^d$, with the corresponding change made in \eqref{phiDef}. 

\subsection{Discussion}

	To the author's knowledge, the results presented in the previous subsection are the first general upper bounds on outside probabilities for {non-autonomous Schr\"odinger operators $H=H_0+V(t)$  with no special structural assumption on $V(t)$.} 	
	This is the  main novelty of the present paper. 	

Our original motivation for considering \eqref{SE} without assuming structural conditions on $V(t)$ is for applications to  nonlinear Schr\"odinger equation of the form \eqref{NLS}. Indeed, according to \thmref{thmNLS}, particle transport is confined to strictly linear light cones for generic non-autonomous, nonlinear Schr\"odinger operator. If we view \eqref{NLS} as the mean-field limit of some quantum many-body dynamical system defined on \textit{bosonic} Fock space over $\ell^2(\Zb^d)$, then this fact is in contrast to the behaviour of the  full many-body system.
% as for the latter superballistic transport is known. In \cite{EG},  Eiser-Gross constructed a 
Indeed, for  bosonic  many-body Hamiltonian, without imposing special localization condition on the initial state, the dynamical spreading rate of particle transport is in general growing logarithmically in $t$.
For the Bose-Hubbard model with nearest neighbour interaction, recent work of Kuwahara-Vu-Saito \cite{KVS} shows particle transport spreads at the rate $t \log t$; see Result 1 therein. For more general long-range bosonic many-body Hamiltonian with hopping matrix satisfying \eqref{Lcond}, together with M.~Lemm and C.~Rubiliani, we proved in \cite{LRZ} that {for initial state with uniformly bounded particle density}, the particle transport spreads at strictly linear rate; see also Faupin-Lemm-Sigal \cites{FLS,FLSa} and our earlier work with Sigal \cites{SZ,LRSZ}. We also note that  explicit example of super ballistic communication (i.e., propagation of general quantum information instead of particle transport) is known; see e.g.~Eisert-Gross \cite{EG}. From this perspective, the general upper bound from \thmref{thm2} is similar  to the classical Lieb-Robinson bound for spin systems \cite{LR} (see also extensions by Hastings-Koma \cite{HK}, Nachtergaele-Sims \cite{NSa} and the review \cite{NS}), in the sense that it establishes general a priori (i.e., state-independent) constraint on many-body quantum dynamics at the mean-field level.

Various types of localization behaviour for discrete nonlinear Schr\"odinger equations, especially in the presence of potential with a large degree of randomness, have been investigated since the seminal works of Fr\"olich-Spencer-Wayne \cite{FSW}; see also \cites{FKS,FKSa,WZ,CSZ}. More recently long-time Anderson localization is proved for deterministic nonlinear Sch\"odinger equation assuming certain quasi-periodicity condition on the potential and weak nonlinearity \cite{CSW}. An excellent review can be found in \cite{FKSb}. In all of these works, the dynamical upper bounds are  tied to the level of  disorderedness and the strength of nonlinearity in question. 
	
	Finally, we note that similar conditions as \eqref{uCond'} were identified in recent works of Liu-Soffer, Soffer-Wu, and Soffer-Stewart \cites{LS,SW,SWa,SSf}  for what the authors call `weakly localized states' in the context of nonlinear Sch\"odinger equations in the continuum, and diffusive bounds such as $\br{\abs{x}}_t\le Ct^{-1/2}$ were established for such states.

\subsection{Approach}\label{secAppr}
Our approach for proving the main result, \thmref{thm1}, originates from the commutator method developed in the classical works of Sigal-Soffer in $N$-body scattering theory \cites{SS,SSa,SSb,SSc} and the propagation observable method from recent works of Faupin-Lemm-Sigal for particle transport problems in condensed matter theory \cites{FLS,FLSa}. It is also a  refinement of the abstract approach developed in  \cite{Zha}.   

We work in the Heisenberg picture and study the Heisenberg evolution, $\al_t$, for time-dependent families of bounded operators $A(t)$, characterized by the duality relation
\begin{align}
	\label{Aevol}
	\inn{u_0}{\al_t(A(t))u_0}=\inn{u_t}{A(t)u_t} ,
\end{align}
where $u_t,\,t\ge0$ is the unique global solution to the Schr\"odinger equation \eqref{SE} with $u_t\vert_{t=0}=u_0\in\cD$. 
The central idea in our approach is to consider the Heisenberg evolution of the characteristic functions $\1_{E_t}$ associated to the outside region
$$E_t:=\Set{x\in\Zb^d:\abs{x}\ge vt^\al+R}.$$  Our  goal is to show the evolution of the time-dependent observables $\al_t(\1_{E_t})$ is upper bounded (as an operator) by its initial state $\1_{E_0}$, up to an $O(t^{-n})$-remainder. Once this is established, estimate \eqref{propGen} follows by evaluating such an operator inequality at the initial state  via the duality relation \eqref{Aevol}.

Whereas directly estimating  the evolution of $\1_{E_t}$ could be difficult, we seek suitable propagation observables $A(t)$ obeying the desired  monotone envelope while being comparable to $\1_{E_t}$. More precisely, we will construct $A(t)$  such that
\begin{itemize}
	\item$\al_t(A(t))$ is constrained by a time-decaying envelope, with $\al(A(t))\le C (A(0)+{t}^{-\beta}),\,\beta>0$.  
	\item $ A(t) $ is comparable with $\1_{E_t}$, in the sense that  $\1_{E_t}\le A(t)$ for $t>0$ while $A(0)\le \1_{E_0}$. 
\end{itemize}
Combining these and using that $\al_t(\cdot)$ is positivity-preserving, we arrive at the desired operator inequality $\al_t(\1_{E_t})\le C (\1_{E_0}+t^{-\beta})$.

\section{Proof of \thmref{thm1}}\label{secPfthm2}
Our proof is based on the \textit{adiabatic spacetime localization observables} (ASTLO, adopting the terminology of \cites{FLS,LRSZ,LRZ}) method, developed in an abstract Hilbert space setting in \cite{Zha}. The ASTLOs play the role of the propagation observables alluded to at the end of the previous section, and we will start with their construction. Next, we derive commutator expansion and monotonicity estimates, leading to a time-decaying envelope for the ASTLOs. Finally, we use the geometric and decay properties of ASTLOs to conclude the desired estimate \eqref{propGen}.

\subsection{Definition and properties of ASTLOs}
%\cjz{Shovel all the geometric properties to here}
Let $\eps>0$. We define two classes of cutoff functions:
\begin{align}
	\label{Ydef}
	\cY:=\Set{w\in C_c^\infty(\Rb)\mid\supp w \subset(0,\eps),w(y)=1\text{ for }\frac\eps4\le y\le\frac{3\eps}4},
\end{align}
and
\begin{align}
	\label{Xdef}
	\cX:=\Set{\chi\in C^\infty(\Rb)\cap L^\infty(\Rb)\mid \chi= c\int_0^xw^2(y)\,dy\text{ for some }w\in\cY,c>0}.
\end{align}
One can view functions in $\cX$ as mollifications of the  Heaviside function $\one_{\Rb>0}$, with transition region in $(0,\eps)$; see \figref{chifig1}. Below we will use functions in $\cX$ with suitable change of variables to construct the propagation observables that control the outsider probabilities.
\begin{figure}[H]
	\centering
	\begin{tikzpicture}[scale=3]
		\draw [->] (-.5,0)--(2,0);
		\node [right] at (2,0) {$\mu$};
		\node [below] at (.3,0) {$0$};
		\draw [fill] (.3,0) circle [radius=0.02];
		
		%					 											\node [below] at (.75,0) {$\eps s$};
		%					 		\draw [fill] (.75,0) circle [radius=0.02];
		
		\node [below] at (1.5,0) {$\eps $};
		\draw [fill] (1.5,0) circle [radius=0.02];

		\draw [very thick] (-.5,0)--(.3,0);
		\draw [very thick] (.3,1)--(2,1);				
		\filldraw [fill=white] (.3,1) circle [radius=0.02];
		
		\draw [dashed, very thick] (-.5,0)--(.75,0) [out=20, in=-160] to (1.5,1)--(2,1);

		\draw [->] (1.55,.5)--(1.3,.5);
		\node [right] at (1.55,.5) {$\chi(\mu)$};
		
		\draw [->] (0,.5)--(.85,.95);
		\draw [->] (0,.5)--(-.05,.05);
		\node [left] at (0,.5) {$\one_{\Rb_{>0}}(\mu)$};
	\end{tikzpicture}
	\caption{A typical function $\chi\in \cX$ with $\norm{\chi}_{\ell^\infty}=1$ compared with the Heaviside function. Here $\eps>0$ is the parameter entering the definition of $\cX$. }\label{chifig1}
\end{figure}
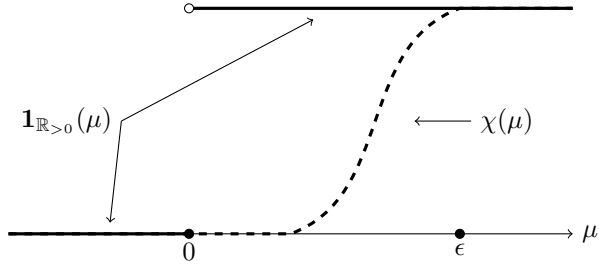

Fix $\bar v>0$ and let $\phi$ be a densely define self-adjoint operator on $\hf$.
For some decay parameter $\al>0$ and adiabatic parameter $s>0$ to be determined later, we define the  {ASTLOs}   as
\begin{equation}\label{Astchi}
	\A_s(t,f):=f\del{\frac{\phi-\bar vt^\al}{s^\al}},\qquad t\ge0,\ f\in L^\infty(\Rb).\tag{ASTLO}
\end{equation} 	
Later on we will mostly choose $\al\in(\frac12,1)$, 
$\phi$ to be the distance function to a ball,  and $s$ to be a linear function of $t$.  But the construction here applies to general self-adjoint operator $\phi$.

%In the next subsections we will prove estimates on $\A_s(t,\chi)$ that are uniform in $s$ and $t$. Then, to make $\A_s(t,\chi)$ comparable with the characteristic functions, we will choose$s$ as functions of $t$ according to the following results: 

The geometric properties of functions in the classes $\cX,\,\cY$ are summarized in the next proposition:
\begin{proposition}
%	[c.f.~\cite{Zha}*{Prop.~3.4}]
	\label{propGeo}
	Let $\eps,\,\bar v>0$ be as in \eqref{Ydef},  \eqref{Astchi}, respectively. 
	Let $\phi$ be the multiplication operator by some $\phi:\Zb^d\to\Rb$.
	 For  any $v>\bar v$ and   $f\in L^\infty(\Rb)$, set
	\begin{align}\label{Adef}
		s(t):=\del{\frac {v-\bar v}{\eps}}^{1/\al} t,\quad 	 \A(t,f):=\A_{s(t)}(t,f)  .
	\end{align}
	Then we have the following:
	\begin{enumerate}
		\item For  any  $\chi\in\cX$, 
			\begin{align}\label{chi-0s-est} %{local-est4}
			\norm{\chi}_{L^\infty}^{-1} \A(0,\chi)  \le&  \1_{\Set{x\in\Zb^d:\phi(x)>0}},\\
			\label{chi-ts-est}\1_{\Set{x\in\Zb^d:\phi(x)>vt^\al}} \le& \norm{\chi}_{L^\infty}^{-1} \A(t,\chi)\qquad(t>0).
		\end{align}
		
		\item Let $w\in C_c^\infty (\Rb)$. Then  $ w\in\cY$ if and only if
			\begin{align}
			\label{wSupp}
			&\supp\A(t,w)\subset\Set{x\in\Zb^d :  \bar v t^\al\le \phi(x)\le  vt^\al},\\
			\label{wConst}
			&\A(t,w)\equiv 1 \text{ on }\Set{x\in\Zb^d :( \tfrac14 v + \tfrac 34\bar v) t^\al\le \phi(x)\le (\tfrac34 v + \tfrac 14\bar v)t^\al} .
		\end{align}
		
	\end{enumerate}

\end{proposition}
\begin{proof}
For Part (1): \eqref{chi-0s-est} follows from the fact that $\supp  \chi(\cdot)\subset (0,\infty)$ for any $\chi\in\cX$.
 For \eqref{chi-ts-est}, we note that if $s(t)$ is given by \eqref{Adef}, then the argument of $[\A(t,f)](\cdot)$ can be written as
\begin{align}
	\label{phiArgRel}
 \phi_{t}(x):=\frac{\phi(x)-\bar v t^\al}{s(t)^\al }=\eps \frac{\phi(x)-\bar vt^\al}{vt^\al-\bar v t^\al}.
\end{align}
 Thus for any $0\le a <b\le1$, we have $\phi_t(x) \in(a\eps,b\eps)$ if and only if $\phi(x)\in([a(v-\bar v)+\bar v]  t^\al,[b(v-\bar v)+\bar v]t^\al)$. Taking $b=1$, we find that  $\phi_t(x)\ge1$ as long as $\phi(x)\ge vt^\al$. Consequently, by definition \eqref{Xdef}, if $\chi\in\cX$ then  $[\A(t,\chi)](x)=\chi(\phi_t(x))\equiv \norm{\chi}_{L^\infty}$ for $\phi(x)\ge vt^\al$. Estimate \eqref{chi-ts-est} follows from here.

For Part (2):  Taking $(a,b)=(0,1)
	,\,(\frac14,\frac34)
	$ in \eqref{phiArgRel} and then comparing with definition \eqref{Ydef} yields the desired results.
\end{proof}

In view of  \propref{propGeo}, one could consider $\A(t,\chi),\chi\in\cX$ as smoothed cutoff functions associated with the outside region $E_t=\Set{x\in\Zb^d:\dist_{B_R}(x)\ge vt^\al}$, and $\A(t,w)$ as those associated with the `boundary' of $E_t$ of width $(v-\bar v)t^\al$.

In the next two subsections we will prove estimates on $\A_s(t,\chi)$ that are uniform in $s$ and $t$. Then, to make $\A_s(t,\chi)$ comparable with the characteristic functions, we will choose $s$ as function of $t$ according to \propref{propGeo}.

\subsection{Symmetrized commutator expansion}
In this subsection we consider commutators between ASTLOs with $\phi$ given by distance function to the balls,
\begin{align}
	\label{phiDef}
\phi=\dist _{B_R} , \quad R\ge0,
\end{align}and  bounded self-adjoint operators $H_0$ with power-law decay; see \eqref{Lcond}. Our estimates will be uniform in $R$, and so we will not display the dependence on $R$ in notations. Changing the definition of $\phi$ to $\dist_X$ with general $X\subset \Zb^d$ yields bounds on outside probabilities with different geometries.

Following the approach of \cite{BFL+} and \cite{Zha}, our goal  is to derive an operator inequality that controls the commutator $i[H_0,\A_s(t,\chi)]$ in terms of totally symmetrized expressions  involving the ASTLOs and the multiple commutators $\ad{k}{\phi}{H_0}$.
Our main result is the following:
\begin{proposition}
	\label{propSymExp}
Let \eqref{Lcond} hold with $n\ge1$. Then for any  $\chi=\int^xw^2\in \cX$,
there exist bounded operators $P_{k}(\chi),k=2,\ldots,n+1$ and $Q_{n+1}(\chi)$ 
	such that the following operator inequality holds on $\hf$ for all $R\ge0$, $\al>0$, and $s>0,\,t\ge0$:
	\begin{equation}
		\label{symExpMain}
		\begin{aligned}
			 &i[ H_0,\A_s(t,\chi)] \le s^{-\al}\A_s(t,w)i[H_0,\phi]\A_s(t,w)+\sum_{k=2}^{n+1} s^{-k\al}P_{k}(\chi)+s^{-(n+1)\al}Q_{n+1}(\chi).
		\end{aligned}
	\end{equation}
%Here   and the $k$-sum is dropped for $n=1$.

Moreover,  the operators  $P_{k}(\chi),\,k=2,\ldots,n+1$ satisfy, for some $\xi_k\in \cX$ depending only on $k,\,\chi$,   
\begin{align}
	\label{PkrhoEst}
	P_{k}(\chi)\le C_{k,\chi}\sbr{1+ M^2} \A_s(t,\xi_k'),
\end{align}
where $M$ is given by \eqref{Lcond}.

Finally, the operator $Q_{m}(\chi)$ is given by $$Q_{m}:= \frac12(R_{m}(\chi)+R_{m}(\chi)^*),$$ with  $R_{m}$ given by \eqref{RmDef}. It satisfies  
\begin{align}\label{QkrhoEst}
{Q_{m}(\chi)}\le C_{m,\chi}M. 
\end{align}
\end{proposition}

\begin{proof} 
	Within this proof we fix $R,\,\al,\,s,\,t$, and all estimates will be independent of these parameters.  For simplicity of notation we also omit the $t$-dependence in notations so that $\A_s(\cdot)\equiv \A_s(t,\cdot)$.
%	Furthermore we take $n\ge2$. For the case $n=1$, terms indexed by $k$ are all dropped and the rest of the proof remains the same. 
	
	1.
Let $\phi$ be given by  \eqref{phiDef} and set
	$$\quad B_{k}:=\ad{k}{\phi}{ iH_0},\quad k=1,\ldots,n+1.$$
	Since   \eqref{Lcond} hold with $n\ge1$ and $\phi=\dist_{B_R}$, by \corref{corA.2}, the $ B_{k}$'s
	are bounded operators  satisfying
	\begin{align}
		\label{BkrhoEst}
		\norm{B_{k}}\le M.
	\end{align}
By \eqref{BkrhoEst}, the assumptions of  \lemref{lemA.2}
 are fulfilled. Therefore it follows from \eqref{535} with $\si=s^\al$ 
 that for $w=(\chi')^{1/2}$  and some $\theta_k\in\cY,\,k=2,\ldots,n+1$,
		\begin{align}
		\label{symExpMain'}
			i[H_0,\A_s(\chi)] \le&   s^{-\al}\A_s(w)B_{1}\A_s(w)+ \sum_{k=2}^{n+1}{{s^{-k\al}}}P_{k}(\chi)+  {s^{-(n+1)\al}} Q_{n+1}(\chi), 
			\end{align}
where, for $k=2,\ldots,n$,  
\begin{align}
	P_{k,\rho}:=&\frac12\del{\frac{1}{(k-1)!}\del{  \A_s(w)^2+ \A_s(w^{(k-1)})B_{k,\rho}^{*}B_{k,\rho}\A_s(w^{(k-1)})}+\frac{1}{k!}\del{(\A_s(\chi^{(k)}))^2+\A_s(\theta_k)B_{k,\rho}^*B_{k,\rho}\A_s(\theta_k)}},\label{Pkdef}\\
	P_{n+1,\rho}:=&\frac12\del{1+\A_s(w)R_n(w)R_n(w)^*\A_s(w)+\sum_{k=2}^n(1+\A_s(\chi^{(k)})R_{n-k+1}(\theta_k)R_{n-k+1}(\theta_k)^*\A_s(\chi^{(k)})}\label{PnDef},\\
	Q_{n+1,\rho}:=&	  R_{n+1,\rho}(\chi)+R_{n+1,\rho}(\chi),\label{Qkdef}
\end{align}
  and $R_{m}$'s are bounded operators defined in \eqref{RmDef}.

%  The assertion \eqref{QkrhoEst} on $Q$ follows from definition \eqref{Qkdef} and estimate  \eqref{RemEst}. 
  
 Since $B_1=i[H_0,\phi]$ in \eqref{symExpMain'}, the leading term is of the desired form. Thus  it remains to 
  prove estimates  \eqref{PkrhoEst}, \eqref{QkrhoEst}. 
  
  2.     We first prove \eqref{PkrhoEst} for $k=2,\ldots,n$. 
Consider  
   the set
$$\cS_k:=\Set{ w,w^{(k-1)},\chi^{(k)},\theta_k}.$$
Take any $\psi\in\hf$ and $v\in\cS_k$. 
By \eqref{BkrhoEst}, we have
\begin{align}
	\label{218}
	\br{\psi, \A_s(v)B_{k}^{*}B_{k}\A_s(v)}=&\norm{B_{k}\A_s(v)\psi}^2\notag\\
	 \le& M^2 \norm{\A_s(v)\psi}^2.
\end{align}
Now we pick $w_k\in\cY$ and $C_{k,\chi}>0$ such that
$$
		C_{k,\chi} w_k\ge v,\quad v\in\cS_k.
$$
Such $w_k$ is easy to construct. For instance, take  $w_k$ s.th.~$w_k=1$ on $\bigcup_{v\in\cS_k}\supp v$ (which is still contained in $(0,\eps)$) and $w_k=0$ away from $(0,\eps)$, and then take $$C_{k,\chi}:=\max_{v\in\cS_k}\norm{v}_{\ell^\infty}.$$ Note this constant depends only on $\chi$ since the functions in $\cS_k$ depend only on  $\chi$ and its derivatives up to $k$-th order. See details in the proof of \lemref{lemA.2}.
%	are all polynomials of functions in $\cY$ and their derivatives.
Choose now $\xi_k\in\cX$ with $$\xi_k\ge C_{k,\chi}^2\int^x w_k^2.$$ By construction, we have
$$\A_s(\xi_k')-\A_s(v)^2=\A_s(w_k)^2-\A_s(v)^2\ge0,\quad v\in\cS_k.$$
This, together with \eqref{218}, implies that for any $v\in\cS_k$,
\begin{align}
	\label{}
	&\br{\psi,[\A_s(v)B_{k}^{*}B_{k}\A_s(v)-D_k^2r^{2\beta k}\A_s(\xi_k')]\psi}\\
	=&\notag
	-D_k^2r^{2\beta k}\br{\psi,[\A_s(\xi_k')-\A_s(v)^2]\psi}+\br{\psi,\A_s(v)\sbr{B_{k}^{*}B_{k}-D_k^2r^{2\beta k}}\A_s(v) \psi}\\
	\le& -D_k^2r^{2\beta k}\br{\psi,[\A_s(\xi_k')-\A_s(v)^2]\psi}\le 0 .\notag
\end{align}
Letting $\psi$ vary in $\hf$, we conclude claim \eqref{PkrhoEst} for $k=2,\ldots,n$.

%  It remains to bound the first term in the r.h.s. of \eqref{symExpMain'} and prove estimate  \eqref{PkrhoEst} for $k=n+1$. 
  
%3. Next, we  bound the first term in the r.h.s.~of \eqref{symExpMain'}. Using the fact 
%	$\br{\psi, \A_s(w)B_{1}\A_s(w)\psi}\le D_1r^{\beta } \norm{\A_s(w)^2\psi}^2$ (see \eqref{suppV} and \lemref{lemA0}) and the relation $w^2=\chi'$, we obtain,  for any $\psi\in\hf$,
%	%
%$$
%		\br{\psi, [\A_s(w)B_{1}\A_s(w)-D_1r^\beta \A_s(\chi')]\psi} 
%		= \br{\psi, ( \A_s(w)\sbr{B_{1}-D_1r^\beta} \A_s(w)) )\psi} \le0.
%$$
%It follows that $\A_s(w)B_{1}\A_s(w)\le D_1r^\beta \A_s(\chi')$.
	
3. To prove \eqref{PkrhoEst} for $P_{n+1}$, we use the remainder estimate \eqref{RemEst} instead of \eqref{BkrhoEst}, and then we proceed exactly as in Step 2 to conclude the desired result. Finally,  the assertion \eqref{QkrhoEst}  follows from definition \eqref{Qkdef} and estimate  \eqref{RemEst}. 
Thus the proof is complete.
\end{proof}

Consider now the Heisenberg evolution, defined by relation \eqref{Aevol} as the dual evolution to \eqref{SE}. Direct computation shows that the evolution of commutative  norm differentiable families of  bounded operators $A(t)$ is governed by the Heisenberg equation
\begin{align}\label{HE}
	\di_t(\al_t(A(t))) = \al_t(D_H(A(t))),
\end{align} 
where $D_H$ denotes the Heisenberg derivative
\begin{align}
	\label{HeisD}&D_HA(t):=\di_tA(t)+i[H, A(t)],\quad H=H_0+ V(t).
\end{align}
Owning to \propref{propSymExp}, we have
\begin{corollary}\label{corHeis}
Let estimate \eqref{symExpMain} hold with $n\ge1$.  Then  
\begin{align}\label{219}
\di_t\al_t\del{ \A_s(t,\chi) }\le& -s^{-\al}\al_t\del{\A_s(t,w)( \bar v \al t^{\al-1}- i[H_0,\phi]) \A_s(t,w)}\notag\\
&+\sum_{n=2}^{n+1} s^{-k\al}\al_t (P_{k}(\chi))+s^{-(n+1)\al}\al_t(Q_{n+1}(\chi)).
\end{align}
% Here the sum is dropped for $n=1$.
\end{corollary}
\begin{proof}
Direct computation shows that $\di_t A_s(t,\chi)=-\bar v \al t^{\al-1}s^{-\al}\A_s(t,\chi')$. This, together with expansion \eqref{symExpMain} and definition \eqref{HeisD}, implies
$$D_H\A_s(t,\chi)\le -s^{-\al}\A_s(t,w)(\bar v \al t^{\al-1}- i[H_0,\phi]) \A_s(t,w)+\sum_{k=2}^{n+1} s^{-k\al}P_{k}(\chi)+s^{-(n+1)\al}Q_{n+1}(\chi).$$
The desired estimate \eqref{219} follows from here, via the Heisenberg equation \eqref{HE} and the positivity-preserving properties of $\al_t(\cdot)$.

\end{proof}

\subsection{Monotone estimates}

Now we use \corref{corHeis} to derive an monotonicity for the ASTLOs. Let $\phi$ be given by \eqref{phiDef}.
For ease of notation, in the remainder of this section we write $$\br{\cdot}_t=\br{u_t,(\cdot)u_t}$$ for a given solution $u_t$ to \eqref{SE}.
Our main result is the following:
	\begin{proposition} \label{propME}
Let estimate \eqref{219} hold with $n\ge1$ and $\frac12<\al\le1$.
Assume moreover  $u_t$ is a solution to \eqref{SE} satisfying, for some $w\in\cY$ and $\delta_1,\,\l>0$,
\begin{align}
	\label{uCond}
\br{\A_s(t,w)i[H_0,\phi]\A_s(t,w)}_t\le (\bar v\al-\delta_1)t^{\al-1}\br{\A_s(t,w)^2}_t,\quad s\ge \l t>0.
\end{align}

Then, for $\chi=\int^x w^2$, 
  there exists  a function  $\eta \in\cX$ (dropped if $n=1$) and some constant $C_4>0$ depending only on $n,\,M$, $\chi$,  $\delta_1$, and $\l$,   such that 
	\begin{align}\label{336}
		 {\br{ (\A_s(t,\chi))}_t}   \le	\br{\A_s(0,\chi)}_0+  {s^{-\g}} \br{\A_s(0,\eta)}_0 +C_4s^{ -\beta},\quad \g:=2\al-1,\ \beta:=\min\del{(n+1)\g,(n+1)\al -1}.
	\end{align}
	(The second term on the r.h.s. is dropped for $n=1$.)
	
\end{proposition}

\begin{proof} 
	
	Within this proof, we fix $s\ge\l t >0$ and a solution $u_t,\,t\ge0$ to \eqref{SE} satisfying \eqref{uCond}.
For ease of notation, for any function $f\in L^\infty$, we write
\begin{align}\label{2.2}
	f[t]:=  \br{ {\A_s(t,f)} }_t .
\end{align}
Thus with this notation the desired estimate \eqref{336} becomes
\begin{align}\label{336'}
	\chi[t]  \le	\chi[0]+  {s^{-\g}} \eta[0]+C_4s^{ -\beta}.
\end{align}

	1.  We start with integrating the differential inequality from \corref{corHeis} to a convenient integral estimate. 
By the Heisenberg equation \eqref{HE} with $A(t)\equiv \A_s(t,\chi)$ and the fundamental theorem of calculus, we find
\begin{align} \label{eq-basic}  
	\chi[t]=\chi[0]+\int_0^t \di_\tau\chi[\tau]\,d\tau .
\end{align}
Applying  estimate \ref{219} and estimates \eqref{PkrhoEst}, \eqref{QkrhoEst} to bound the integrand in \eqref{eq-basic}, we obtain, for some $C_1=C_1(n,\chi,M)>0$,
\begin{align}
	\label{225}
	\di_\tau\chi[\tau]\le s^{-\al}(-\bar v\al t^{\al-1} \chi'[\tau]+\br{\A_s(\tau,w)i[H_0,\phi]\A_s(\tau,w)}_\tau) +C_1\sum_{k=2}^{n+1} s^{-k\al}\xi_k'[\tau]+C_1s^{-(n+1)\al}. 
\end{align}
We now show that the leading term is bounded by a strictly negative multiple of $\chi'[\tau]$. 
First, using the choice $\chi'= w^2$ and   condition \eqref{uCond}, we have
\begin{align*}
	\label{}
	-\bar v\al \tau^{\al-1} \chi'[\tau]+\br{\A_s(\tau,w)i[H_0,\phi]\A_s(\tau,w)}_\tau=& \br{\A_s(\tau,w)\sbr{-\bar v\al t^{\al-1}+i[H_0,\phi]}\A_s(\tau,w)}_\tau\\
	=& -\delta_1 \tau^{\al-1}\br{\A_s(\tau,w)^2}=-\delta_1 \tau^{\al-1}\chi'[\tau].
\end{align*}
Next, for $s\ge \l t>0$, we have for all $\tau\le t$  that $\tau^{1-\al}\le \l^{\al-1}s^{1-\al}$,  and therefore
$$\frac{\tau^{\al-1}}{s^\al}=\frac{1}{s^\al \tau^{1-\al}}\ge \frac{1}{\l ^{\al-1}s}.$$
Combining the last two inequalities, we conclude that the leading term in \eqref{225} can be bounded as 
\begin{align}
	\label{leadEst}
	s^{-\al}(-\bar v\al t^{\al-1} \chi'[\tau]+\br{\A_s(\tau,w)i[H_0,\phi]\A_s(\tau,w)}_\tau)\le - \frac{\delta_1}{\l^{\al-1}s}\chi'[\tau].
\end{align}

For the sum in \eqref{225}, we choose a function $\xi\in\cY$ such that $\xi'\ge C_1\sup_{k}\xi_k'$ pointwise. (Such function is easy to choose, see Step 2 in the proof of \propref{propSymExp} for a construction.) Then it follows 
\begin{align}
	\label{sumEst}
	C_1\sum_{k=2}^{n+1} s^{-k\al}\xi_k'[\tau]\le s^{-2\al}\xi'[\tau]. 
\end{align}

Combining \eqref{leadEst} and \eqref{sumEst} in \eqref{225} yields
\begin{align}
	\label{225'}
	\di_\tau\chi[\tau]\le - \frac{\delta_1}{\l^{\al-1}s} \chi'[\tau]+ s^{-2\al} \xi'[\tau]+C_1s^{-(n+1)\al}.
\end{align}
Note importantly that the second term in the r.h.s.~ is of lower order if and only if $\al>1/2$.

Plugging \eqref{225'} back to \eqref{eq-basic},  transposing the $O(s^{-1})$ term to the left, and using $t\le \l^{-1}s$ to bound the remainder,  we arrive at the integral estimate
\begin{align}
	\label{intEst1}
	\chi[t]+  \frac{\delta_1}{\l^{\al-1}s}\int_0^t\chi'[\tau]\,d\tau \le \chi[0] +s^{-2\al} \int_0^t \xi'[\tau]\,d\tau +C_1\l^{-1}s^{1-(n+1)\al}.
\end{align}

2. Now we claim the following holds: There exist $\tilde \chi\in \cX$ (dropped for $n=1$) and $C_2>0$ depending only on $n,\,\chi,\,M,$   and $\bar \delta :=\delta_1 \min(\l^{1-\al}, \l^{2-\al})$
such that for
$$\g=2\al-1,$$
which is strictly positive for $\al>1/2$, we have
\begin{align}
	&\int_0^t \chi'[t]dr  \le  \bar \delta ^{-1}s\chi[0]
		+   C_2 s^{1-\g}\tilde \chi[0]+    C_2 s^{1-n\g}  .
	\label{propag-est31} 
\end{align}
Here the second term on the r.h.s.~is dropped for $n=1$. 
Furthermore,   for fixed $n,\chi,M$, we have
\begin{align}
	\label{C2Est}
	C_2=O(\bar \delta^{-n} ) ,\quad \bar\delta\to0. 
\end{align}
%\jz{Note that for $\al>1/2$, the second and the third terms on the r.h.s.~is of lower order for large $s$ compared to the first one.}

We prove \eqref{propag-est31} using \eqref{intEst1} by induction in $n$. 
First, we drop the first term in the r.h.s.~of \eqref{intEst1} (which is non-negative) and then multiplying both side by $s \delta_1^{-1}\l^{\al-1}$   to get
\begin{align}
	&\int_0^t \chi'[t]dr  \le  {s\delta_1^{-1}\l^{\al-1}s\chi[0]
		+  \delta_1^{-1}\l^{\al-1}s^{-\g}\int_0^t \xi'[\tau]\,d\tau+    C_1\delta_1^{-1}\l^{\al-2}s^{2-(n+1)\al} }.
	\label{intEst2} 
\end{align}
 
For the base case $n=1$, we note that $\xi$ satisfies the uniform bound $\xi'\le C_3$ for some $C_3=C_3(n,M,\chi)>0$,   since $\xi$ depends only on $C_1$ and $\xi_k$ (see \eqref{sumEst}) and 
 $\xi'\in C_c^\infty$ by definition \eqref{Xdef}). Thus the integrand on the r.h.s.~of \eqref{intEst2} can be bounded as $\int_0^t\xi'[\tau]\,d\tau \le C_3\l^{-1}s$ for $0< t\le \l^{-1}s$. Plugging this back to \eqref{intEst2} and setting $C_2=\bar \delta^{-1}(C_1+C_3)$
yields \eqref{propag-est31}--\eqref{C2Est} for $n=1$.

Next, assuming \eqref{propag-est31} holds for $n-1$ with $n\ge2$, we prove it for $n$. Indeed, by induction hypothesis, the integral on the r.h.s.~of \eqref{intEst2} is bounded as
\begin{align}
	\label{}
	\int_0^t \xi'[\tau]\,d\tau\le  C_{2,n-1} \del{s\xi[0]
		+  s^{1-\g}\tilde \xi[0]+    s^{1-(n-1)\g} }.
\end{align}
Here $C_{2,n-1}=O(\bar \delta ^{n-1})$ and the $\tilde\xi$-term is dropped for $n=2$. 
Plugging this back to \eqref{intEst2} yields
\begin{align}
	&\int_0^t \chi'[t]dr  \le     {\delta_1^{-1}\l^{\al-1}s\chi[0]
		+  \delta_1^{-1}\l^{\al-1}C_{2,n-1}  \del{s^{1-\g}\xi[0]
			+  s^{1-2\g}\tilde \xi[0]+    s^{1-n\g} }+    C_1\delta_1^{-1}\l^{\al-2} s^{2-\al(n+1)} }.
	\label{intEst3} 
\end{align}
{By assumption $\al>1/2$, we have $\g>0$ and so the third term on the r.h.s.~is of lower order compared to the second one. } Moreover, for $0<\al \le 1$ we have $1-n\g \ge 2-\al(n+1)$ for all $n\ge1$. Thus the last term is of lower order compared to the penultimate one. Using these facts and choosing $$C_{2,n}=\bar \delta^{-1} (C_{2,n-1}+C_1),$$ and  any majorant $\tilde \chi\in \cX$ such that %
$$
	\tilde \chi (\mu)\ge \xi (\mu),\,\tilde \xi(\mu),\quad \mu\in\Rb,
$$we conclude the desired estimates \eqref{propag-est31}--\eqref{C2Est} for $n$ from \eqref{intEst3}. This completes the induction and \eqref{propag-est31} is proved.

	3. Finally, we use \eqref{propag-est31} to derive the desire estimate \eqref{336}.

	Dropping the second term in the l.h.s. of \eqref{intEst1}, which is non-negative since  $\delta_1>0$ and $ \chi'[\tau]\ge0$ for all $r$, we obtain
	\begin{align} \label{128}
				&\chi[t]   
		\le   \chi[0]+  s^{-2\al}\int_0^t  \xi'[\tau]\,d\tau +  C_1\l^{-1}{  s^{1-(n+1)\al}}.
%		\chi[t]\le  \chi[0]+ \sum_{k=2}^ns^{-k}\int_0^t  \xi_k'[\tau]\,d\tau  + C{ t s^{-(n+1)}},
	\end{align}
	We apply estimate \eqref{propag-est31} to the  second term in the r.h.s.~of \eqref{128}, yielding
	\begin{align}\label{e3442}
		\chi[t]\le  \chi[0]+ \bar \delta ^{-1}s^{-\g}\xi[0]
		+   C_2 s^{-2\g}\tilde \xi[0]+    C_2 s^{-(n+1)\g}+ C_1\l^{-1}{  s^{1-(n+1)\al}}.
	\end{align}
	Here the third term is dropped for $n=1$. 
	Take now %
	\begin{align}
		\label{C4def}
		C_4:=2C_2+\l^{-1}C_1,
	\end{align} and $\eta\in\cX$ satisfying  $\eta (\mu)\ge \bar\delta^{-1}\xi (\mu),\,C_2\tilde \xi(\mu)$  pointwise (drop $\tilde \xi$ for $n=1$). Then   we conclude the desired estimate  \eqref{336'}  from \eqref{e3442}. This completes the proof of \propref{propME}.
\end{proof}

\subsection{Completing the proof of \thmref{thm1}}
Let $u_t$ be a solution to \eqref{SE} satisfying \eqref{uCond'} for some $W$.  
Since $\bar v = v-\frac{\delta}{2\al},$  we have $v>\bar v$ and $\bar v \al-\frac\delta2 =v\al -\delta$. Moreover, owning to \propref{propGeo}, Part (2), there exists $w\in\cY$ s.th.~$\A(t,w)\equiv \A_{s(t)}(t,w)  =W_t(\phi)$ for all $t>0$, $s(t)=\del{\frac {v-\bar v}{\eps}}^{1/\al} t$, and
$\phi=\dist_{B_R}$. 
Using these facts, together with assumption \eqref{uCond'} and the Cauchy-Schwarz inequality, we find for all $t>0$ that
\begin{align*}
	\label{}
	\br{\A(t,w)i[H_0,\phi]\A(t,w)}_t=& \br{W_t(\phi)i[H_0,\phi]W_t(\phi)}_t\\
	\le&\norm{W_t(\phi)u_t}\norm{i[H_0,\phi]W_t(\phi)u_t}\\
	\le&(  v\al-\delta)t^{\al-1}\norm{W_t(\phi)u_t}^2\\
	=&(\bar v \al - \frac\delta2)t^{\al-1}\br{ {\A(t,w)^2 } }_t.
\end{align*}
From here we conclude that 
condition \eqref{uCond} is satisfied with $\delta_1 =\frac\delta2$, $\l=\del{\frac {v-\bar v}{\eps}}^{1/\al}$, and the choice $s=s(t)$.

Take now $\chi=\int^xw^2\in\cX$. Applying \propref{propME} to $\chi$, we get a function $\eta\in\cX$ and some constant $C_4>0$ s.th.~\eqref{336} holds. Furthermore, since $s$ is given by \eqref{Adef} and $\phi=\dist_{B_R}$, we have by \propref{propGeo}, Part (1) that 
\begin{align*}
				\norm{\chi}_{L^\infty}^{-1} \A(0,\chi)  \le&  \1_{\Set{\abs{x}>R}},\\
 \1_{\Set{\abs{x}>R+vt^\al}} \le& \norm{\xi}_{L^\infty}^{-1}\A(t,\xi)\qquad(\xi=\chi,\eta,\ t>0).
\end{align*}
Inserting these bounds into \eqref{336}, multiplying both sides by $\norm{\chi}_{L^\infty}$, and recalling the choice that $s=\l t$, we arrive at, for $C_4>0$ as in \eqref{C4def},
\begin{align}
	\label{pen}
\br{ \1_{\Set{\abs{x}>R+vt^\al}}}_t   \le \br{\1_{\Set{\abs{x}>R}}}_0+  \frac{\norm{\chi}_{L^\infty}}{\norm{\eta}_{L^\infty}}\l^{-\g}{t^{-\g}} \br{\1_{\Set{\abs{x}>R}}}_0 +C_4\norm{\chi}_{L^\infty}\l^{-\beta}t^{ -\beta}	.
\end{align}
Lastly, choosing  
\begin{align}
	\label{Cdef}
	C= \frac{\norm{\chi}_{L^\infty}}{\norm{\eta}_{L^\infty}}\l^{-\g}+C_4\norm\chi\l^{-\beta},
\end{align}   and noting that 
$\br{ \1_{\Set{\abs{x}>N}}}_t=P(N,t)$, we conclude the desired estimate \eqref{propGen} from \eqref{pen}. 
 This completes the proof of \thmref{thm1}.

\section{Proof of \thmref{thm3}}\label{secPfBall}

%Fix   $p\ge0$ and an $\ell^2$-solution  $q(\cdot)$ to \eqref{NLS}. For some $v>0$ to be determined, we  decompose $H_p(q(t))$ into interior and exterior regions as
%%
%\begin{align}
%	\label{HpDecomp}
%	H_p	(q(t))= I (t)+E(t),
%\end{align}
%where
%\begin{align}
%	\label{Idef}
%	I (t)= 	\sum _{ \abs{j}\le  2vt}\abs{j}^p \abs{q_j(t)}^2,\\
%	E (t)= \sum _{ \abs{j}>   2vt} \abs{j}^p\abs{q_j(t)}^2.\label{Edef}
%\end{align}
%To bound  $I_R(t)$, we use the trivial estimate %
%\begin{align}
%	\label{IestFin}
%	I(t)\le C_{d,p,v}   t   ^{dp} ,
%\end{align} which is valid as long as $\abs{q_j(t)}\le 1$.
%%It follows that 
%%\begin{align}
%%	\label{IestFin}
%%	 \frac{1}{T^{dp+1}}\int_0^T I_R(t)\le C_{d,v}. 
%%\end{align}
%
%Thus it remains to bound the contribution from $E(t)$. To this end, 

To prove \eqref{MainEst}, we first derive  an $\ell^2$-ballistic transport bound that controls the dynamical spreading of wavepackets  at various scales.
Then we use the $\ell^2$-bound together with a dyadic multiscale argument to   conclude  uniform bound for $P_r(N,t)$.

%\subsection{$\ell^2$-ballistic upper bound } \label{sec2.1}

Recall the outside probability $P(N,t)$ is defined in \eqref{Pdef}. 
Using \thmref{thm2}, we derive the following:
\begin{proposition}
	Assume \eqref{propEst'} holds with some $n\ge1$. Then 	for any  $A-1\ge B>0$ and $v>\kappa$, the solution   $u(x,t)$ to \eqref{SE}  satisfies, for all $t\ge1$,
	\begin{equation}\label{propEst}
		\sum_{\abs{x}\ge Avt}	\abs{u(x,t)}^2
		\le C _{M,v,A-B}\del{\del{\sum_{\abs{x}\ge Bvt }	\abs{u(x,0)}^2} + {t}^{-n} }.
	\end{equation}  
\end{proposition}
\begin{proof}
	
	Since \eqref{propEst'} holds uniformly in $R>0$ and $v>\kappa$, applying it with  $R=Bvt$ and $v'=(A-B)v>\kappa$  yields the desired estimate \eqref{propEst}.
\end{proof}
Next, we lift the $\ell^2$-ballistic upper bound \eqref{propEst} to higher moments via a multiscale argument. 
%Our main result is:
%\begin{proposition}
%	Let the assumptions of \thmref{thm3} hold. Then for $v>\kappa+\delta$, $E(t)$ from \eqref{Idef}, and all $t\ge1$,
%	\begin{align}
%		\label{EestFin}
%		E(t)\le& C_{d,p,n,v,u_0,s}. 
%	\end{align}
%\end{proposition} 
%\begin{proof}
Fix $v$, $r$ and write, for $P_r$ from \eqref{HpDef},
	\begin{align}
		\label{112}
		P_r(2vt,t)= \sum_{k=1}^\infty Q_k(t),\quad Q_k(t):=\sum _{2^kvt \le \abs{x}\le 2^{k+1}t}\abs{x}^r \abs{u(x,t)}^2. 
	\end{align} 
	For each $k=1,2,\ldots$, we have  
	\begin{align}
		\label{18}
		Q_k(t)\le&C_{d,r,v}(2^{k+1}t)^{dr}\sum _{\abs{x}\ge 2^kvt}  \abs{u(x,t)}^2.
	\end{align}
	Since $v>\kappa$, we apply \eqref{propEst} with $A=2^k$, $B=2^{k-1}$,  to bound the sum in \eqref{18}   as
	\begin{align}
		\label{19'}
		\sum _{\abs{x}\ge 2^kvt}  \abs{u(x,t)}^2 \le & C_{M,v} \del{\del{\sum_{\abs{x}\ge 2^{k-1}vt }	\abs{u(x,0)}^2}^{1/2}+ {((2^{k}-2^{k-1})t)}^{-n} }^2.
	\end{align}
	
To bound the first term on the r.h.s., we use condition \eqref{q0Suff}. Indeed, for any $\rho>0$, we have 	$P_{r_0}(0,0)\ge \sum_{\abs{x}\ge \rho} \abs{x}^{r_0}\abs{u_0(x)}^2\ge\rho ^{dr_0} \sum_{\abs{x}\ge \rho}  \abs{u_0(x)}^2$. Transposing this and using \eqref{q0Suff} gives
	\begin{align}
		\label{q0Cond}
		\sum_{\abs{x}\ge \rho} \abs{u(x,0)}^2\le C_{u_0}\rho^{-dr_0},\quad \rho>0.
	\end{align}
	Using \eqref{q0Cond} with $\rho=2^{k-1}vt$ to bound the first term in the r.h.s.~of \eqref{19'}, we find
	\begin{align}
		\label{19}
		\sum _{\abs{x}\ge 2^kvt}  \abs{u(x,t)}^2  \le & C_{M,v,u_0}\sbr{ (2^{k-1}t)^{-dr_0/2}+ (  2^{k-1} t)^{-n}}^2 .
	\end{align}
	Combining \eqref{18} with \eqref{19}, we conclude that 
	\begin{align}
		\label{115}
		Q_k(t)\le& C_{d,r,v, M,u_0} (2^{k+1}t)^{dr}\sbr{ (2^{k-1}t)^{-dr_0/2}+ (  2^{k-1} t)^{-n}}^2\notag\\
		\le&  C_{d,r,v, M,u_0,\g}(2^{k+1}t)^{-\g},\qquad \g =\min (dr_0,2n)-dr.
	\end{align}
	Owning to \eqref{115}, for any  $ t\ge 1$, we have %
	\begin{align}
		\label{QkEst}
		Q_k(t)\le C_{d,r,v, M,u_0,\g}2^{-\g k}.
	\end{align} Plugging \eqref{QkEst} into \eqref{112} yields
	\begin{align}
		\label{EestFin'}
			P_r(2vt,t)\le& C_{d,r,v, M,u_0,\g}\sum_{k=1}^\infty 2^{-\g k} .
	\end{align}
		By the assumptions $r_0>r$ and $n>dr_0/2$, we have   $\g>0$ in \eqref{EestFin'}. Thus the series in \eqref{EestFin'} converges and so \eqref{MainEst} follows.
%\end{proof}

\section{Proof of \thmref{thmNLS}}\label{secPfNLS}
Fix any $\ell^2$-solution, say $q^{(0)}$, to \eqref{NLS}, and consider the frozen-coefficient equation
\begin{align}
	\label{NLSf}
	i\di_t u =  (H_0+ V(t))u,\quad   V(t)= \bar V(t)+\abs{N(t,q^{(0)}t)}.
\end{align}
By the assumptions on $\bar V(t)$ and $N(t,q)$, the potential $V(t)$ is uniformly bounded for all $t$. Thus \thmref{thm1} applies to the equation \eqref{NLSf}. Since $q^{(0)}$ satisfies \eqref{NLS}, it is also a solution to  \eqref{NLSf}. Thus the conclusion of Thms.~\ref{thm2}--\ref{thm3} hold for $q^{(0)}$. 

Now observe that for any multiplication operator by $\phi:\Zb^d\to \Rb$, we have the identity
\begin{align}
	\label{cLrel}
	[	  V(t),\phi ]\equiv 0.
\end{align}
Since the proofs of Thms.~\ref{thm2}--\ref{thm3} are based on commutator estimates with functions of $\phi$, the conclusions are independent of $V(t)$ in \eqref{NLSf}  owning to identity \eqref{cLrel}.
Indeed, the commutator estimate \eqref{symExpMain} remains valid as-is if $H_0$ is changed to $H_0+V(t)$, whereas in all subsequent results, the dynamics enter only through  \eqref{symExpMain}. 
Therefore the estimates in Thms.~\ref{thm2}--\ref{thm3}  are uniform in  $q^{(0)}$. Letting $q^{(0)}$ vary then yields \thmref{thmNLS}.

\begin{remark}
	\label{remNLSprop}
	Note that the same argument cannot be extended to the continuum. Indeed, consider \eqref{NLS} with  $H_0=-\Lap$, $\bar V=0$, and $N=\abs{q}^{p-1},p>1$ on $\Rb^d.$ Then \thmref{thm2} does not hold without any energy cutoff adapted to $H_0$, since the momentum operator $i[H_0,\abs{x}]$ is unbounded. However,  commutator estimates on $H_0+V(t)$ with such a cutoff, say $g(H_0)$ with $g\in C_c^\infty(\Rb)$, would be dependent on the frozen solution $q^{(0)}$, since $[N(q^{(0)}),g(H_0)]\ne0$. Thus in the continuous case, the equality \eqref{cLrel} cannot hold in general. Indeed, as it turns out, propagation estimate for nonlinear Schr\"odinger  equation in the continuum requires a lot more delicate treatment than describe above; see e.g.~\cites{AFPS,LS,SSf,SW,SWa}. 
\end{remark}

\section*{Acknowledgments}

The Author is supported by National Key R \& D Program of China Grant 2022YFA100740, China Postdoctoral Science Foundation Grant 2024T170453, and the Shuimu Scholar program of Tsinghua University. The Author thanks M.~Lemm and Wencai Liu for helpful remarks. He thanks M.~Lemm, S.~Rademacher, C.~Rubiliani, and I.~M.~Sigal  for fruitful collaborations leading to this work.

\section*{Declarations}
\begin{itemize}
	\item Conflict of interest: The Author has no conflicts of interest to declare that are relevant to the content of this article.
	\item Data availability: Data sharing is not applicable to this article as no datasets were generated or analysed during the current study.
\end{itemize}

\appendix

\section{Commutator estimates}\label{sec:A}

In this appendix, we prove some technical lemmas for deriving commutator expansion \propref{propSymExp}. 
Most of the results can be found in \cite{Zha}  but are adapted here for the reader's convenience.
%sufficient condition for  uniform estimates on multiple commutators with (multiplication operator by) Lipschitz functions. 
%We work in the continuum but the same proof works on $\Zb^d$. 
%In particular, this implies that \eqref{commHphi} holds uniformly with $H_0$ satisfying  \eqref{Lcond} and $\phi$ given by \textit{any} distance function.
%In the remainder of this section we fix $R\ge0$ in definition \eqref{phiRhodef}. All estimates below are independent of $R$. 

\begin{lemma}
	\label{lemA0}
	Let $H_0$ satisfy \eqref{Lcond} with $n\ge1$ and let $\phi(x)$ be a Lipschitz function satisfying $\abs{\phi(x)-\phi(y)}\le L\abs{x-y}$ for some $L>0$ and all $x,\,y\in\Zb^d$. Then 
	\begin{align}
		\label{uEst}
		\norm{{\ad{k}{\phi}{H_0}} }\le  {M_k L^{  k}} ,\quad   M_k:=\sup_{x\in\Zb^d}\sum_{y\in\Zb^d}\abs{H_0(x,y)} \abs{x-y}^k.
	\end{align}
\end{lemma}
\begin{proof}
For any multiplication operators by $\phi:\Zb^d\to \Cb$ and all $k=1,\,2,\ldots$, a straightforward induction shows for  any $u\in\hf$,
\begin{equation}\label{A.2}
	\sbr{\ad{k}{\phi}{H_0}}u(x)=\sum_{y\in\supp u}(\phi(y)-\phi(x))^kH_0(x,y)u(y) .
\end{equation}
By the Lipschitz property of $\phi(x)$, we have
$$\abs{\phi(x)-\phi(y)}^k\le    {L^k}  \abs{x-y}^k.$$
Plugging this back to \eqref{A.2}, we find  
\begin{align}
	\label{A7}
	\abs{\sbr{\ad{k}{\phi}{H_0}}u(x)}\le  L^k \sum _{y\in \Zb^d} \abs{x-y}^k\abs{H_0(x,y)}\abs{u(y)}.
\end{align}
Using Schur's rest, we conclude \eqref{uEst} from \eqref{A7}.
\end{proof}

\begin{corollary}\label{corA.2}
	Let $H_0$ satisfy \eqref{Lcond} with $n\ge1$. Then for every $X\subset \Rb^d$, the distance function $ \dist_X(x)$ satisfies
	$$
	\norm{\ad{k}{\dist_X}{H}}\le M\qquad (1\le k\le n+1).
	$$
\end{corollary}
\begin{proof}
	All $d_X$ satisfies the Lipschitz bound with $L=1$ and $M\ge M_k$ for $k=1,\ldots,n+1$. 
\end{proof}

The next lemma is a refinement of previous commutator expansion results in \cite{BFL+} and \cite{Zha}; see also \cite{HSS}*{Lem.~2.1}. 
Let  $\phi$ be a densely defined self-adjoint operator defined on $\hf\equiv \ell^2(\Zb^d),\,d\ge1$.
Fix $t\in\Rb$, $\l>0$   and write
$$	\A(f):=f\del{\frac{\phi-\bar vt^\al}{\si}},\qquad   f\in L^\infty(\Rb).$$
 We prove the following:
\begin{lemma}\label{lemA.2}
Let $H_0$ be a bounded self-adjoint operator on $\hf$. Assume the operator $\phi$ satisfies
\begin{align}
	\label{commCond}
	\norm{B_k}<\infty \quad \text{ where }\ B_k:=i\ad{k}{\phi}{H_0},\ k=1,\ldots,n+1.
\end{align}
Then, for any $\chi=\int^x w^2\in\cX,\,w\in\cY$ (see \eqref{Ydef}--\eqref{Xdef}),  there exist $\theta_k\in\cY$, $k=2,\ldots,n$ (dropped if $n=1$),
	%	 each of which satisfying $$\supp\xi_k'\subset \supp\chi',$$
%	together with a constant  $C=C(n,\vec\k,\chi)>0,$  
	such that the following operator inequality holds on $\hf$:
\begin{align} 
	&i[H_0,\A(\chi)]\label{535} \\\le&\notag   \si^{-1}\A(w)B_1\A(w)\\&+ \sum_{k=2}^{n}{\frac{\si^{-k}}{2}}\sbr{\frac{1}{(k-1)!}\del{  \A(w)^2+ \A(w^{(k-1)})B_{k}^{*}B_{k}\A(w^{(k-1)})}+\frac{1}{k!}\del{(\A(\chi^{(k)}))^2+\A(\theta_k)B_k^*B_k\A(\theta_k)}}
	\notag\\
	&+ \frac{\si^{-(n+1)}}{2}\del{1+\A(w)R_n(w)R_n(w)^*\A(w)+\sum_{k=2}^n(1+\A(\chi^{(k)})R_{n-k+1}(\theta_k)R_{n-k+1}(\theta_k)^*\A(\chi^{(k)})}\notag\\
	&+\frac{\si^{-(n+1)}}{2} \del{R_{n+1}(\chi)+ R_{n+1}(\chi)^*}.\notag
\end{align}
Here the $k$-sums are dropped for $n=1$, and $R_{m}$ is a bounded operator given by
\begin{align}
	\label{RmDef}
	R_m(\chi)=\int d\tilde \chi(z)(z-\phi)^{-m}B_m(z-\phi)^{-1},
\end{align}
where $d\tilde \chi(z)$ is a measure on $\Cb$ satisfying $d\tilde \chi (z)\equiv 0$ for $\Im z=0$ and 
\begin{align}
	\label{dchiEst}
	\int\abs{d\tilde\chi(z)}\abs{\Im(z)}^{-(p+1)}\le C_{p,\chi},\quad p\ge 1.
\end{align}

\end{lemma}
\begin{proof}
1. We start with a classical commutator expansion,
stated in sufficiently general form in \cite{BFL+}*{Appd.~B}:
Let $\chi$ be a real function with  $\chi'\in C_c^\infty$. Under commutator bound \eqref{commCond}, there holds
\begin{align} i [H_0, \A(\chi)]=&
	\sum_{k=1}^n \frac{\si^{- k}}{k!}\A(\chi^{(k)}) B_{k} + \si^{-(n+1)}R_{n+1}(\chi)\label{commex'}, 
	%	\\
	%	= &\sum_{k=1}^n(-1)^{k-1}\frac{\si^{- k}}{k!}B_{k}\A(\chi^{(k)})  - \si^{-(n+1)}R_{n+1}^*(\chi)	\label{commex}
\end{align}
where $R_m$ is given by \eqref{RmDef} with $d\tilde \chi(z)$ satisfying \eqref{dchiEst}. In particular, we have by \eqref{dchiEst} that 
\begin{align}\label{RemEst}
	\norm{R_m(\chi)}\le C_{m,\chi}\norm{B_{m}}.
\end{align}

Adding commutator expansion \eqref{commex'} to its adjoint and dividing the result by two, we obtain
\begin{align}
	i[H_0,\A(\chi)]=&\mathrm{I}+\mathrm{II}+\mathrm{III},\label{gvExp}\\
	\mathrm{I}	=&	\frac{1}{2}\si^{-1} \left(\A(\chi')B_1+B_1^*\A(\chi')\right),\label{leading}
	\\	\mathrm{II}=&\frac{1}{2}\sum_{k=2}^{n}\frac{\si^{-k}}{k!}\left(\A(\chi^{(k)})B_k+B_k^*\A(\chi^{(k)})\right),\label{A-sym-exp}
	\\\mathrm{III}=&{\frac{1}{2}\si^{-(n+1)}\left(R_{n+1}(\chi)+R_{n+1}(\chi)^*\right)}\label{Rem3},
\end{align}
where the term $\mathrm{II}$ is dropped for $n=1$.
%\cjz{How to bound the stand alone remainder term?? 
	%
	%8.28 owning to this, perhaps it's safe to leave all of $R$'s without young's inequa.
	%
	%This seems unbounded without cutoff. And the cutoff will not be gone in general
	%} 

2. We first bound the term $\mathrm{I}$ in line \eqref{leading}.  Let $w:=\sqrt{\chi'}$, which is well defined and lies in $C_c^\infty(\Rb)$ by definition \eqref{Xdef}. Then expansion \eqref{commex'} holds for $u$. This expansion, together with the facts that $B_1=i[H,\phi]=B_1^*$, $\Ad_{\phi}^l(B_k)=B_{k+l}$, implies
\begin{align}
	&\quad	\A(\chi')B_1+B_1^*\A(\chi')\notag
	\\	&=\A(w)^2B_1+B_1\A(w)^2\notag\\&=2\A(w)B_1\A(w)+\A(w)[\A(w),B_1]+[B_1,\A(w)]\A(w)\nonumber\\
	&=2\A(w)B_1\A(w)\notag\\&\quad+\sum_{l=1}^{n-1}\frac{\si^{-l}}{l!}\left( {\A(w)B_{1+l} \A(w^{(l)})+\A(w^{(l)})B_{1+l}^{*}\A(w)}\right)\label{518}\\&\quad+\si^{-n}(\A(w)R_n(w)+R_n(w)^{*}\A(w)),\label{519}
\end{align}
where  line \eqref{518} is dropped for $n=1$.
% and,  for some 	$C=C(n,\vec\k,\chi)>0$,
%\begin{align}\label{Rem1Est}
%	\norm{\Rem_1}\le C.
%\end{align}

We will bound the terms in \eqref{518}--\eqref{519} using the operator estimate
\begin{align}\label{op-esti}
	P^*Q+Q^*P&\leq P^*P+Q^*Q. 
\end{align}
For the terms in line \eqref{518}, we use \eqref{op-esti} with
$$
P=\A(w),\quad Q:= B_{1+l}\A(w^{(l)}) , 
$$
yielding 
\begin{align}\label{Exp1}
	&\si^{-l}({\A(w)B_{1+l} \A(w^{(l)})+\A(w^{(l)})B_{1+l}^{*}\A(w)})\notag\\\le& \si^{-l}\del{  \A(w)^2+ \A(w^{(l)})B_{1+l}^{*}B_{1+l}(\A(w^{(l)}))^2}.
\end{align}
For the remainder terms in \eqref{519}, we apply \eqref{op-esti} with $$
P= 1,\quad Q=R_n(w)^*\A(w),
$$ to obtain
\begin{align}\label{Exp2}
	\si^{-n}(\A(w)R_n(w)^*+R_n(w)^*\A(w))&\leq \si^{-n}\del{1+\A(w)R_n(w)R_n(w)^*\A(w)}.
\end{align}
Combining \eqref{Exp1} and \eqref{Exp2} in \eqref{leading} and writing $k=l+1$ yields 
\begin{align}\label{Exp4}
	&\mathrm{I}\le  \si^{-1} \A(w)B_1\A(w)\\&\quad+\frac12\sum_{k=2}^{n}\frac{\si^{-k}}{(k-1)!}\del{  \A(w)^2+ \A(w^{(k-1)})B_{k}^{*}B_{k}\A(w^{(k-1)})}+\frac12\si^{-(n+1)}\del{1+\A(w)R_n(w)R_n(w)^*\A(w)}.\notag
\end{align}
This bound the term $\mathrm{I}$ \eqref{leading}.

%	Since $u$, $u^{(l)}$ and $\theta^1$ are supported in $(0,\delta)$, we can modify $\xi^2,...,\xi^n$ in such a way that $\xi^l\in\cX$ majorizes $u^2$, $\tilde{\theta}^2$ and $(u^{(l)})^{2}$ for each $l=1,...,n-1$.
%	
%	Since $\A(u^2)=\A(u)^2$ 

3. For $n\ge2$, the term $\mathrm{II}$ in line \eqref{A-sym-exp}  is bounded similarly as in Step 2. For $k=2,...,n$, we take $\theta_k\in\cY$ (see \eqref{Ydef}) with
\begin{align}
	\label{thetaCond}
 \theta_k\equiv 1 \text{ on }\supp\chi^{(k)}.
\end{align}
Indeed, this is possible since $\supp\chi'\subset (0,\eps)$ by \eqref{Xdef}
%		by taking $P=\chi_\si^{(k)}$ and $Q=\theta_\si^kB_k$, where $\theta_k_s\equiv \theta_k(\si^{-1}(\Phi-vt))$, in \eqref{op-esti}, where \jz{each $\theta^{k}\in C_c^{\infty}((0,\delta/2))$} is a smooth \sout{characteristic} \jz{cutoff} function that takes value $1$ on $\supp(\chi^{(k)})$, we have
%\begin{align}\label{opera-ineq-Bk}
%	\pm \si^{-k}(\chi_\si^{(k)}B_k+B_k^*\chi_\si^{(k)})&\leq \si^{-k}\left((\chi_\si^{(k)})^2+\|B_k\|^2(\theta_\si^k)^2\right).
%\end{align}

We first show that 
\begin{align}\label{527}
	\A(\chi^{(k)})B_k=\A(\chi^{(k)})B_{k}\A(\theta_k)+ \si^{-(n+1-k)}\Rem_k.
\end{align}
By relation \eqref{thetaCond}, we have $\chi^{(k)}=\chi^{(k)}\theta_k=\theta_k\chi^{(k)}$. Using this, commutator expansion \eqref{commex'},  and the fact that $\Ad_{\phi}^l(B_k)=B_{k+l}$, we find
\begin{align}
	&\quad\A(\chi^{(k)})B_k\notag\\&=\A(\chi^{(k)})\A(\theta_k)B_k\notag\\&=\A(\chi^{(k)})B_{k}\A(\theta_k)+\A(\chi^{(k)})[\A(\theta_k),B_k]\nonumber\\
	&=\A(\chi^{(k)})B_k\A(\theta_k)\notag\\&\quad+\sum_{l=1}^{n-k}\frac{\si^{-l}}{l!}\A(\chi^{(k)})\A(\theta_k^{(l)})B_{k+l}+\si^{-(n+1-k)}\A(\chi^{(k)})R_{n-k+1}(\theta_k),\label{531}
\end{align}
where the $l$-sum is dropped for $k=n$ 
%	and \begin{align}\label{RemkEst}
	%		\Rem_k\le C,\quad k=2,\ldots,n,
	%	\end{align}
%	for some 	$C=C(n,\vec\k,\chi)>0$.

Since $\theta_k\equiv 1$ on $\supp(\chi^{(k)})$ by \eqref{thetaCond}, we have $\supp(\theta_k^{(l)})\cap\supp(\chi^{(k)})=\emptyset$ for all $l\geq 1$. Therefore, in line \eqref{531}, 
\begin{align*}
	\A(\chi^{(k)}) \A(\theta_k^{(l)})B_{k+l}=0,\quad l=1,\ldots n-k.
\end{align*}
Thus the $l$-sum in \eqref{531} is dropped, and 
estimate \eqref{527} follows from here. 

Adding \eqref{527} to its adjoint yields
\begin{align}\label{526}
	&\si^{-k}\del{\A(\chi^{(k)})B_k+B_k^*\A(\chi^{(k)})}\\=&\notag \si^{-k}\del{\A(\chi^{(k)})B_k\A(\theta_k)+\A(\theta_k)B_k^*\A(\chi^{(k)})}+\si^{-(n+1)}(\A(\chi^{(k)})R_{n-k+1}(\theta_k)+R_{n-k+1}(\theta_k)^*\A(\chi^{(k)})).
\end{align}
We apply estimate \eqref{op-esti} on the r.h.s. of \eqref{526} as in Step 1, and then sum over $k$ to obtain
\begin{align}\label{Exp3}
	\mathrm{II}	\leq&\frac12\sum_{k=2}^{n} \frac{\si^{-k}}{k!}\left((\A(\chi^{(k)}))^2+\A(\theta_k)B_k^*B_k\A(\theta_k)\right)\notag\\&+\frac12\si^{-(n+1)}\sum_k(1+\A(\chi^{(k)})R_{n-k+1}(\theta_k)R_{n-k+1}(\theta_k)^*\A(\chi^{(k)})).
\end{align}
This bounds the term  $\mathrm{II}$ in line \eqref{A-sym-exp}.

%Since $n$ is finite, we can choose $\xi^2,...,\xi^n\in \cX$ such that $(\xi^k)'$ majorizes $(\chi^{(k)})_\si^2+\|B_k\|^2(\theta_k_s)^2$ for each $k$. 

4. 
%To bound the term $\mathrm{III}$ in  \eqref{Rem3}, we proceed similarly as in \cite{bibid}
Plugging \eqref{Exp4}, \eqref{Exp3}, and \eqref{Rem3}  back to \eqref{gvExp} 
%and using bounds \eqref{Bcond}, \eqref{RemEstEst}, \eqref{Rem1Est}, and \eqref{RemkEst}, 
we conclude the desired inequality \eqref{535}.

\end{proof}

\bibliography{DynLocNLSBib}
\end{document}